\documentclass[11pt, onecolumn]{IEEEtran}
\usepackage{amsmath}
\usepackage{amssymb}
\usepackage{verbatim}
\usepackage{epsfig}

\newtheorem{theorem}{Theorem}
\newtheorem{corollary}{Corollary}

\newcommand{\snr}{\rm SNR}

\begin{document}
\title{Fractional Power Control for \\ Decentralized Wireless Networks}

\author{Nihar Jindal, Steven Weber, Jeffrey G. Andrews
\thanks{The contact author N. Jindal (nihar@umn.edu) is with the University of Minnesota,
S. Weber is with Drexel University, J. Andrews is with the
University of Texas at Austin.  This research was supported by NSF
grant no. 0634763 (Jindal), no. 0635003 (Weber), nos. 0634979 and
0643508 (Andrews), and the DARPA IT-MANET program, grant no.
W911NF-07-1-0028 (all authors). An early, shorter version of this
work appeared at \emph{Allerton} 2007 \cite{JinWeb_Allerton07}.
Manuscript date: \today.}}

\maketitle

\begin{abstract}
We consider a new approach to power control in decentralized
wireless networks, termed fractional power control (FPC).
Transmission power is chosen as the current channel quality raised
to an exponent $-s$, where $s$ is a constant between 0 and 1. The
choices $s = 1$ and $s = 0$ correspond to the familiar cases of
channel inversion and constant power transmission, respectively.
Choosing $s \in (0,1)$ allows all intermediate policies between
these two extremes to be evaluated, and we see that usually neither
extreme is ideal. We derive closed-form approximations for the
outage probability relative to a target SINR in a decentralized (ad
hoc or unlicensed) network as well as for the resulting transmission
capacity, which is the number of users/m$^2$ that can achieve this
SINR on average. Using these approximations, which are quite
accurate over typical system parameter values, we prove that using
an exponent of $s^*=\frac{1}{2}$ minimizes the outage probability,
meaning that the inverse square root of the channel strength is a
sensible transmit power scaling for networks with a relatively low
density of interferers. We also show numerically that this choice of
$s$ is robust to a wide range of variations in the network
parameters. Intuitively, $s^*=\frac{1}{2}$ balances between helping
disadvantaged users while making sure they do not flood the network
with interference.
\end{abstract}

\section{\label{sec:1} Introduction}

Power control is a fundamental adaptation mechanism in wireless
networks, and is used to at least some extent in virtually all
terrestrial wireless systems. For a single user fading channel in
which the objective is to maximize expected rate, it is optimal to
increase transmission power (and rate) as a function of the
instantaneous channel quality according to the well-known
waterfilling policy \cite{GolVar97}. On the other hand, if the
objective is to consistently achieve a target rate (or SNR), then
the power should be adjusted so that this target level is exactly
met.  Such an objective is philosophically the opposite of
waterfilling, since power is inversely related to the instantaneous
channel quality: we call this \emph{channel inversion}. Although
suboptimal from an information theory point of view, some channel
inversion is used in many modern wireless systems to adapt to the
extreme dynamic range (often $> 50$ dB due to path loss differences
as well as multipath fading) that those systems experience, to
provide a baseline user experience over a long-term time-scale.

\subsection{Background and Motivation for Fractional Power Control}

In a multi-user network in which users mutually interfere, power
control can be used to adjust transmit power levels so that all
users simultaneously can achieve their target SINR levels.  The
Foschini-Miljanic algorithm is an iterative, distributed power
control method that performs this task assuming that each receiver
tracks its instantaneous SINR and feeds back power adjustments to
its transmitter \cite{FosMil93}.  Considerable work has deeply
explored the properties of these algorithms, including developing a
framework that describes all power control problems of this type
\cite{Yat95}, as well as studying the feasibility and implementation
of such algorithms \cite{BamChe95, HerCho00}, including with varying
channels \cite{ChaVee03}; see the recent monographs
\cite{SchubertBocheNOW}\cite{MChiangNOW} for excellent surveys of
the vast body of literature.  This body of work, while in many
respects quite general, has been primarily focused on the cellular
wireless communications architecture, particularly in which all
users have a common receiver (i.e., the uplink). More recently,
there has been considerable interest in power control for
decentralized wireless networks, such as unlicensed spectrum access
and ad hoc networks
\cite{ElbEph02,CruSan03,Hae03,AgaKri01,KawKum03,Chi05}.  A key
distinguishing trait of a decentralized network is that users
transmit to distinct receivers in the same geographic area, which
causes the power control properties to change considerably.

In this paper, we explore the optimal power control policy for a
multi-user decentralized wireless network with mutually interfering
users and a common target SINR.  We do not consider iterative
algorithms and their convergence.  Rather, motivated by the poor
performance of channel inversion in decentralized networks
\cite{WebAndJin07}, we develop a new transmit power policy called
\emph{fractional power control}, which is neither channel inversion
nor fixed transmit power, but rather a trade-off between them.
Motivated by a recent Motorola proposal \cite{XiaRat06} for fairness
in cellular networks, we consider a policy where if $H$ is the
channel power between the transmitter and receiver, a transmission
power of $H^{-s}$ is used, where $s$ is chosen in $[0,1]$. Clearly,
$s=0$ implies constant transmit power, whereas $s=1$ is channel
inversion.  The natural question then is: what is an appropriate
choice of $s$?  We presume that $s$ is decided offline and that all
users in the network utilize the same $s$.

\subsection{Technical Approach}

We consider a spatially distributed (decentralized) network,
representing either a wireless ad hoc network or unlicensed
spectrum usage by many nodes (e.g., Wi-Fi or spectrum sharing
systems).  We consider a network that has the following key
characteristics.
\begin{itemize}
    \item Each transmitter communicates with a single receiver that is
a distance $d$ meters away.
    \item Channel attenuation is determined by path loss (with
    exponent $\alpha$) and a (flat) fading value $H$.
    \item Each transmitter knows the channel power to its intended
    receiver, but has no knowledge about other transmissions.
    \item All multi-user interference is treated as noise.
    \item Transmitters do not schedule their transmissions based on their channel conditions or the activities of other nodes.
    \item Transmitter node locations are modeled by a homogeneous spatial (2-D) Poisson process.
\end{itemize}
These modeling assumptions are made to simplify the analysis, but
in general reasonably model a decentralized wireless network with
random transmitter locations, and limited feedback mechanisms. In
particular, the above assumptions refer to the situation where a
connection has been established between a transmitter and
receiver, in which case the channel power can be learned quickly
either through reciprocity or a few bits of feedback.  It is not
however as easy to learn the interference level since it may
change suddenly as interferers turn on and off or physically move
(and reciprocity does not help).  The fixed transmit distance
assumption is admittedly somewhat artificial, but is significantly
easier to handle analytically, and has been shown to preserve the
integrity of conclusions even with random transmit distances. For
example, \cite{WebAndJin07,WebYan05} prove that picking the
source-destination distance $d$ from an arbitrary random
distribution reduces the transmission capacity by a constant
factor of $E[d^2]/(E[d])^2 \geq 1$.  Therefore, although fixed
distance $d$ can be considered best-case as far as the numerical
value of transmission capacity, this constant factor will not
change fractional power control's relative effect on the
transmission capacity, which is the subject of this paper.

\subsection{Contributions and Organization}

The contributions of the paper are the suggestion of fractional
power control for wireless networks and the derivation of the
optimum power control exponent $s^*=\frac{1}{2}$.  The exponent
$s=\frac{1}{2}$ is shown to be optimal for an approximation to the
outage probability/transmission that is valid for relatively low
density networks that are primarily interference-limited (i.e., the
effect of thermal noise is not overly large); if the relative
density or the effect of noise is large, then our numerical results
show that no power control ($s=0$) is generally preferred. In the
relatively large parameter space where our primary approximation is
valid, fractional power control with the choice $s^* = \frac{1}{2}$
is shown to greatly increase the transmission capacity of a 1-hop ad
hoc network for small path loss exponents (as $\alpha \to 2$), with
more modest gains for higher attenuation channels. The results open
a number of possible avenues for future work in the area of power
control, and considering the prevalence of power control in
practice, carry several design implications.

The remainder of the paper is organized as follows.  Section II
provides background material on the system model, and key prior
results on transmission capacity that are utilized in this paper.
Section III holds the main results, namely Theorem \ref{thm:3} which
gives the outage probability and transmission capacity achieved by
fractional power control, and Theorem \ref{thm:4} which determines
the optimum power control exponent $s^*$ for the outage probability
approximation. Section IV provides numerical plots that explore the
numerically computed optimal $s^*$, which provides insight on how to
choose $s$ in a real wireless network. Section V suggests possible
extensions and applications of fractional power control, while
Section VI concludes the paper.

\section{\label{sec:2} Preliminaries}

\subsection{System Model}

We consider a set of transmitting nodes at an arbitrary snapshot in
time with locations specified by a homogeneous Poisson point process
(PPP), $\Pi(\lambda)$, of intensity $\lambda$ on the infinite
two-dimensional plane, $\mathbb{R}^2$.  We consider a reference
transmitter-receiver pair, where the reference receiver, assigned
index $0$, is located without loss of generality, at the origin. Let
$X_i$ denote the distance of the $i$-th transmitting node to the
reference receiver.  Each transmitter has an associated receiver
that is assumed to be located a fixed distance $d$ meters away.  Let
$H_{i0}$ denote the (random) distance--independent fading
coefficient for the channel separating transmitter $i$ and the
reference receiver at the origin; let $H_{ii}$ denote the (random)
distance--independent fading coefficient for the channel separating
transmitter $i$ from its intended receiver.  We assume that all the
$H_{ij}$ are i.i.d. (including $i=j$), which implies that no
source-destination (S-D) pair has both a transmitter and receiver
that are very close (less than a wavelength) to one another, which
is reasonable.  Received power is modelled by the product of
transmission power, pathloss (with exponent $\alpha
> 2$), and a fading coefficient. Therefore, the (random) SINR at
the reference receiver is:
\begin{eqnarray}
{\rm SINR}_0 = \frac{P_0 H_{00} d^{-\alpha}}{\sum_{i \in
\Pi(\lambda)} P_i H_{i0} X_i^{-\alpha} + \eta},
\end{eqnarray}
where $\eta$ is the noise power.  Recall our assumption that
transmitters have knowledge of the channel condition, $H_{ii}$,
connecting it with its intended receiver.  By exploiting this
knowledge, the transmission power, $P_i$, may depend upon the
channel, $H_{ii}$.  If Gaussian signaling is used, the corresponding
achievable rate (per unit bandwidth) is $\log_2 ( 1 + {\rm
SINR}_0)$. The Poisson model requires that nodes decide to transmit
independently, which corresponds in the above model to slotted ALOHA
\cite{Bac06}. A good scheduling algorithm by definition introduces
correlation into the set of transmitting nodes, which is therefore
not well modeled by a homogeneous PPP. We discuss the implications
of scheduling later in the paper.

\subsection{Transmission Capacity}

In the outage-based transmission capacity framework, an outage
occurs whenever the SINR falls below a prescribed threshold $\beta$,
or equivalently whenever the instantaneous mutual information falls
below $\log_2(1+ \beta)$.  Therefore, the system-wide outage
probability is
\begin{equation}
q(\lambda) = \mathbb{P}({\rm SINR}_0 < \beta)
\label{eq:opdef}
\end{equation}
Because (\ref{eq:opdef}) is computed over the distribution of
transmitter positions as well as the iid fading coefficients (and
consequently transmission powers), it corresponds to fading that
occurs on a time-scale that is comparable or slower than the
packet duration (if (\ref{eq:opdef}) is to correspond roughly to
the packet error rate).  The outage probability is clearly a
continuous increasing function of the intensity $\lambda$.

Define $\lambda(\epsilon)$ as the maximum intensity of
\textit{attempted} transmissions such that the outage probability is
no larger than $\epsilon$, i.e., $\lambda(\epsilon)$ is the unique
solution of $q(\lambda) = \epsilon$.  The transmission capacity is
then defined as $c(\epsilon) = \lambda(\epsilon) (1 - \epsilon) b$,
which is the maximum density of \textit{successful} transmissions
times the spectral efficiency $b$ of each transmission. In other
words, transmission capacity is area spectral efficiency subject to
an outage constraint.

For the sake of clarity, we define the constants $\delta = 2/\alpha
< 1$ and ${\rm SNR}=\frac{p d^{-\alpha}}{\eta}$. Now consider a
path-loss only environment ($H_{i0} = 1$ for all $i$) with constant
transmission power ($P_i = p$ for all $i$).  The main result of
\cite{WebYan05} is given in the following theorem.

\begin{theorem}[\cite{WebYan05}] \label{thm:1}
{\bf Pure pathloss.} {\em Consider a network where the SINR at the
reference receiver is given by (\ref{eq:opdef}) with $H_{i0} = 1$
and $P_i = p$ for all $i$.  Then the following expressions give
bounds on the outage probability and transmission attempt
intensity for $\lambda, ~ \epsilon$ small:
\begin{eqnarray}
q^{\rm pl}(\lambda) &\geq& q^{\rm pl}_l(\lambda) = 1 - \exp
\left\{- \lambda \pi d^2
\left(\frac{1}{\beta } - \frac{1}{\snr} \right)^{-\delta} \right\}, \label{eq-pout} \\
\lambda^{\rm pl}(\epsilon) & \leq & \lambda^{\rm pl}_u(\epsilon) =
-\log(1-\epsilon)\frac{1}{\pi d^2} \left(\frac{1}{\beta } -
\frac{1}{\snr} \right)^{\delta}. \label{eq-transcap}
\end{eqnarray}}
\end{theorem}

Here {\em pl} denotes pathloss.  The transmission attempt
intensity upper bound, $\lambda^{\rm pl}_u(\epsilon)$, is obtained
by solving $q^{\rm pl}_l(\lambda) = \epsilon$ for $\lambda$. These
bounds are shown to be approximations for small $\lambda,\epsilon$
respectively, which is the usual regime of interest.  Note also
that $- \log(1-\epsilon) = \epsilon+ O(\epsilon^2)$, which implies
that transmission density is approximately linear with the desired
outage level, $\epsilon$, for small outages.  The following
corollary illustrates the simplification of the above results when
the noise may be ignored.

\begin{corollary}
\label{cor:1} {\em When $\eta = 0$ the expressions in Theorem
\ref{thm:1} simplify to:
\begin{eqnarray}
\label{eqn:a}
q^{\rm pl}(\lambda) &\geq& q^{\rm pl}_l(\lambda) = 1 - \exp \left\{- \lambda \pi d^2 \beta^{\delta} \right\}, \\
\lambda^{\rm pl}(\epsilon) & \leq & \lambda^{\rm pl}_u(\epsilon) =
-\log(1-\epsilon)\frac{1}{\pi d^2 \beta^{\delta}}. \label{eqn:a2}
\end{eqnarray}}
\end{corollary}

\section{\label{sec:3} Fractional Power Control}

The goal of the paper is to determine the effect that fractional
power control has on the outage probability lower bound in
(\ref{eq-pout}) and hence the transmission capacity upper bound in
(\ref{eq-transcap}).  We first review the key prior result that we
will use, then derive the maximum transmission densities $\lambda$
for different power control policies.  We conclude the section by
finding the optimal power control exponent $s$.

\subsection{Transmission capacity under constant power and channel inversion}

In this subsection we restrict our attention to two well-known
power control strategies: constant transmit power (or no power
control) and channel inversion.  Under constant power, $P_i = p$
for all $i$ for some common power level $p$.  Under channel
inversion, $P_i = \frac{p}{\mathbb{E}[H^{-1}]}H_{ii}^{-1}$ for all
$i$. This means that the received signal power is $P_i H_{ii}
d^{-\alpha} = \frac{p}{\mathbb{E}[H^{-1}]} d^{-\alpha}$, which is
constant for all $i$.  That is, channel inversion compensates for
the random channel fluctuations between each transmitter and its
intended receiver. Moreover, the expected transmission power is
$\mathbb{E}[P_i] = p$, so that the constant power and channel
inversion schemes use the same expected power.  We would like to
emphasize the distribution of $H$ is arbitrary and can be adapted
in principle to any relevant fading or compound shadowing-fading
model.  For some possible distributions (such as Rayleigh fading,
i.e. $H \sim \exp(1)$), the value $\mathbb{E}[H^{-1}]$ may be
undefined, strictly speaking.  In practice, the transmit power is
finite and so $P_i = \frac{p}{\mathbb{E}[H^{-1}]}H_{ii}^{-1}$ is
finite.  The value $\mathbb{E}[H^{-1}]$ is simply a normalizing
factor and can be interpreted mathematically to mean that $H \to
\min(H,\delta)$ for an arbitrarily small $\delta$. Such a
definition would not affect the results in the paper.

A main result of \cite{WebAndJin07} extended to include thermal
noise is given in the following theorem, with a general proof that
will apply to all three cases of interest: constant power, channel
inversion and fractional power control. Note that {\rm cp} and
${\rm ci}$ are used to denote constant power and channel
inversion, respectively.

\begin{theorem}\label{thm:2}
{\em {\bf Constant power.} Consider a network where the SINR at the
reference receiver is given by (\ref{eq:opdef}) with $P_i = p$ for
all $i$.  Then the following expressions give good approximations of
the outage probability and transmission attempt intensity for
$\lambda,\epsilon$ small.
\begin{eqnarray} \nonumber
q^{\rm cp}(\lambda) &\geq& q^{\rm cp}_l(\lambda) = 1 - \mathbb{P}
\left(H_{00} \geq \frac{\beta}{\snr}\right)
 \mathbb{E} \left[ \exp \left\{-\lambda \pi d^2
\mathbb{E}[H^{\delta}] \left(\frac{H_{00}}{\beta} - \frac{1}{\snr}
\right)^{-\delta}
\right\} \Big\vert H_{00} \geq \frac{\beta}{\snr} \right] \\
& \approx & \tilde{q}^{\rm cp}_l(\lambda) = 1 - \mathbb{P}
\left(H_{00} \geq \frac{\beta}{\snr}\right) \exp \left\{-\lambda \pi
d^2 \mathbb{E}[H^{\delta}] \mathbb{E} \left[
\left(\frac{H_{00}}{\beta} - \frac{1}{\snr} \right)^{-\delta}
\Big\vert
H_{00} \geq \frac{\beta}{\snr} \right] \right\} \nonumber \\
\lambda^{\rm cp}(\epsilon) & \approx & \tilde{\lambda}^{\rm
cp}(\epsilon) = - \log\left( \frac{1-\epsilon}{\mathbb{P}
\left(H_{00} \geq \frac{\beta}{\snr}\right)}  \right) \frac{1}{\pi
d^2} \frac{1}{\mathbb{E}[H^{\delta}]} \mathbb{E}
 \left[ \left(\frac{H_{00}}{\beta}   - \frac{1}{\snr}
 \right)^{-\delta} \Big\vert H_{00} \geq \frac{\beta}{\snr}
 \right]^{-1}. \label{eq:cp3}
\end{eqnarray}
{\bf Channel inversion.} Consider the same network with $P_i =
\frac{p}{\mathbb{E}[H^{-1}]}H_{ii}^{-1}$ for all $i$. Then the
following expressions give tight bounds on the outage probability
and transmission attempt intensity for $\lambda,\epsilon$ small:
\begin{eqnarray} \label{eq:ci1}
q^{\rm ci}(\lambda) &\geq& q^{\rm ci}_l(\lambda) = 1 - \exp
\left\{-\lambda \pi d^2
\mathbb{E}[H^{\delta}]\mathbb{E}[H^{-\delta}] \left(\frac{1}{\beta}
- \frac{\mathbb{E}[H^{-1}]}{\snr }
\right)^{-\delta} \right\}  \\
\lambda^{\rm ci}(\epsilon) & \leq & \lambda^{\rm ci}_u(\epsilon) = -
\log(1-\epsilon) \frac{1}{\pi d^2}
\frac{1}{\mathbb{E}[H^{\delta}]\mathbb{E}[H^{-\delta}]}
\left(\frac{1}{\beta} - \frac{\mathbb{E}[H^{-1}]}{\snr } \right)^{\delta} . \label{eq:ci2} \\
\nonumber
\end{eqnarray}}
\end{theorem}

\begin{proof}
The SINR at the reference receiver for a generic power vector
$\{P_i\}$ is
\begin{equation}
{\rm SINR}_0 = \frac{P_0 H_{00} d^{-\alpha}}{\sum_{i \in
\Pi(\lambda)} P_i H_{i0} X_i^{-\alpha} + \eta},
\end{equation}
and the corresponding outage probability is
\begin{equation}
q(\lambda) = \mathbb{P}({\rm SINR}_0 < \beta) = \mathbb{P} \left(
\frac{P_0 H_{00} d^{-\alpha}}{\sum_{i \in \Pi(\lambda)} P_i H_{i0}
X_i^{-\alpha} + \eta} < \beta \right).
\end{equation}
Rearranging yields:
\begin{equation}
q(\lambda) = \mathbb{P} \left( \sum_{i \in \Pi(\lambda)} P_i H_{i0}
X_i^{-\alpha} \geq \frac{P_0 H_{00} d^{-\alpha}}{\beta} - \eta
\right).
\end{equation}
Note that outage is certain when $P_0 H_{00} < \eta \beta
d^{\alpha}$. Conditioning on $P_0 H_{00}$ and using $f(\cdot)$ to
denote the density of $P_0 H_{00}$ yields:
\begin{equation}
q(\lambda) = \mathbb{P} \left(P_0 H_{00} \leq \eta \beta d^{\alpha}
\right) + \int_{\eta \beta d^{\alpha}}^{\infty} \mathbb{P} \left(
\sum_{i \in \Pi(\lambda)} P_i H_{i0} X_i^{-\alpha} \geq \frac{p_0
h_{00} }{\beta d^{\alpha}} - \eta ~ \Big\vert ~ P_0 H_{00} = p_0
h_{00} \right) f(p_0 h_{00}) {\rm d} (p_0 h_{00}).
\end{equation}
Recall the generic lower bound from \cite{WebAndJin07}: if
$\Pi(\lambda) = \{(X_i,Z_i)\}$ is a homogeneous marked Poisson point
process with points $\{X_i\}$ of intensity $\lambda$ and iid marks
$\{Z_i\}$ independent of the $\{X_i\}$, then
\begin{equation}
\mathbb{P} \left(\sum_{i \in \Pi(\lambda)} Z_i X_i^{-\alpha} > y
\right) \geq 1 - \exp \left\{ - \pi \lambda \mathbb{E}[Z^\delta]
y^{-\delta} \right\},
\end{equation}
Applying here with $Z_i = P_i H_{i0}$ and $y = \frac{p_0 h_{00}
}{\beta d^{\alpha}} - \eta$:
\begin{eqnarray}
q(\lambda) & \geq & \mathbb{P} \left(P_0 H_{00} \leq \eta \beta d^{\alpha} \right) + \int_{\eta \beta d^{\alpha}}^{\infty} \left(1 - \exp \left\{- \pi \lambda \mathbb{E}[(P_i H_{i0})^{\delta}] \left( \frac{p_0 h_{00} }{\beta d^{\alpha}} - \eta \right)^{-\delta} \right\} \right) f(p_0 h_{00}) {\rm d} (p_0 h_{00}) \nonumber \\
&=& 1 - \int_{\eta \beta d^{\alpha}}^{\infty} \exp \left\{- \pi \lambda \mathbb{E}[(P_i H_{i0})^{\delta}] \left( \frac{p_0 h_{00} }{\beta d^{\alpha}} - \eta \right)^{-\delta} \right\} f(p_0 h_{00}) {\rm d} (p_0 h_{00}) \nonumber \\
&=& 1 - \mathbb{P} \left(P_0 H_{00} \geq \eta \beta d^{\alpha}
\right) \mathbb{E} \left[ \exp \left\{- \lambda \pi d^2
\mathbb{E}[(P_i H_{i0})^{\delta}] \left( \frac{P_0 H_{00} }{\beta }
- \frac{\eta}{d^{-\alpha}} \right)^{-\delta} \right\} \Big\vert P_0
H_{00} \geq \eta \beta d^{\alpha} \right]. \label{eq-outage_lower}
\end{eqnarray}
The Jensen approximation for this quantity is:
\begin{equation}
q(\lambda) \approx 1 - \mathbb{P} \left(P_0 H_{00} \geq \eta \beta
d^{\alpha} \right) \exp \left\{- \lambda \pi d^2 \mathbb{E}[(P_i
H_{i0})^{\delta}] \mathbb{E} \left[ \left( \frac{P_0 H_{00} }{\beta
} - \frac{\eta}{d^{-\alpha}} \right)^{-\delta} \Big\vert P_0 H_{00}
\geq \eta \beta d^{\alpha} \right] \right\}.
\label{eq-outage_jensen}
\end{equation}

For constant power we substitute $P_i H_{i0} = p H_{i0}$ (for all
$i$) into (\ref{eq-outage_lower}) and (\ref{eq-outage_jensen}) and
manipulate to get the expressions for $q^{\rm cp}_l(\lambda)$ and
$\tilde{q}^{\rm cp}_l(\lambda)$ in (\ref{eq:cp3}).  To obtain
$\tilde{\lambda}^{\rm cp}(\epsilon)$,  we solve $\tilde{q}^{\rm
cp}_l(\lambda) = \epsilon$ for $\lambda$. For channel inversion,
$P_0 H_{00} = \frac{p}{\mathbb{E}[H^{-1}]}$ while for $i \ne 0$  we
have $P_i H_{i0}
=\frac{p}{\mathbb{E}[H^{-1}]}\frac{H_{i0}}{H_{ii}}$. Plugging into
(\ref{eq-outage_lower}) and using the fact that $H_{ii}$ and
$H_{i0}$ are i.i.d. yields (\ref{eq:ci1}), and (\ref{eq:ci2}) is
simply the inverse of (\ref{eq:ci1}).
\end{proof}
Note that  channel inversion only makes sense when $\frac{\snr
}{\mathbb{E}[H^{-1}]} = \frac{p d^{-\alpha}}{\eta
\mathbb{E}[H^{-1}]}$, the effective interference-free SNR after
taking into account the power cost of inversion, is larger than the
SINR threshold $\beta$.
 The validity of the outage lower bound/density upper bound as well
 as of the Jensen's approximation
are evaluated in the numerical and simulation results in Section
\ref{sec:numerical}.

When the thermal noise can be ignored, these results simplify to the
expressions given in the following corollary:
\begin{corollary}
\label{cor:2} {\em When $\eta = 0$ the expressions in Theorem
\ref{thm:2} simplify to:
\begin{eqnarray}
\label{eqn:b}
q^{\rm cp}(\lambda) &\geq& q^{\rm cp}_l(\lambda) = 1 - \mathbb{E} \left[ \exp \left\{-\lambda \pi d^2 \beta^{\delta} \mathbb{E}\left[H^{\delta} \right] H_{00}^{-\delta}\right\}\right] \nonumber \\
& \approx & \tilde{q}^{\rm cp}_l(\lambda) = 1 - \exp \left\{-\lambda \pi d^2 \beta^{\delta} \mathbb{E}\left[H^{\delta} \right] \mathbb{E} \left[H^{-\delta} \right] \right\}, \nonumber \\
q^{\rm ci}(\lambda) & \geq & q^{\rm ci}_l(\lambda) = 1 - \exp \left\{-\lambda \pi d^2 \beta^{\delta} \mathbb{E}\left[H^{\delta} \right] \mathbb{E} \left[H^{-\delta} \right] \right\}, \nonumber \\
\lambda^{\rm cp}(\epsilon) & \approx & \tilde{\lambda}^{\rm cp}(\epsilon) = - \log(1-\epsilon) \frac{1}{\pi d^2 \beta^{\delta}} \frac{1}{\mathbb{E}\left[H^{\delta} \right] \mathbb{E} \left[H^{-\delta} \right]}, \nonumber \\
\lambda^{\rm ci}(\epsilon) & \leq & \lambda^{\rm ci}_u(\epsilon) = - \log(1-\epsilon) \frac{1}{\pi d^2 \beta^{\delta}} \frac{1}{\mathbb{E}\left[H^{\delta} \right] \mathbb{E} \left[H^{-\delta} \right]}.
\end{eqnarray} }
\end{corollary}
Note that these expressions match Theorem 3 and Corollary 3 of the
SIR-analysis performed in \cite{WebAndJin07}.

In the absence of noise the constant power outage probability
approximation equals the
 channel inversion outage probability lower bound: $\tilde{q}_l^{\rm cp}(\lambda) = q_l^{\rm ci}(\lambda)$.
As a result, the constant power transmission attempt intensity
approximation equals the channel inversion
 transmission attempt intensity upper bound: $\tilde{\lambda}^{\rm cp}(\epsilon) = \lambda_u^{\rm ci}(\epsilon)$.
 Comparing $\tilde{\lambda}^{\rm cp}(\epsilon) = \lambda_u^{\rm ci}(\epsilon)$ in (\ref{eqn:b}) with
 $\lambda_u^{\rm pl}(\epsilon)$ in (\ref{eqn:a2}) it is evident that the impact of fading on the transmission capacity is measured by the loss factor, $L^{\rm cp} = L^{\rm ci}$, defined as
\begin{equation}
L^{\rm cp} = L^{\rm ci} = \frac{1}{\mathbb{E}\left[H^{\delta} \right] \mathbb{E} \left[H^{-\delta} \right]} < 1.
\end{equation}
The inequality is obtained by applying Jensen's inequality to the convex function $1/x$ and the random variable $H^{\delta}$.  If constant power is used, the $\mathbb{E}[H^{-\delta} ]$ term is due to fading of the desired signal while the $\mathbb{E}[H^{\delta}]$ term is due to fading of the interfering links.  Fading of the interfering signal has a positive effect while fading of the desired signal has a negative effect.  If channel inversion is performed the $\mathbb{E}[H^{-\delta} ]$ term is due to each interfering transmitter using power proportional to $H_{ii}^{-1}$.  When the path loss exponent, $\alpha$, is close to 2 then $\delta = 2/\alpha$ is close to one, so the term $\mathbb{E}[H^{-\delta} ]$ is nearly equal to the expectation of the inverse of the fading, which can be extremely large for severe fading distributions such as Rayleigh.  As a less severe example, $\alpha =3$, the loss factor for Rayleigh fading is $L^{\mathrm{cp}} = L^{\mathrm{ci}} = 0.41$.

\subsection{Transmission capacity under fractional power control}

In this section we generalize the results of Theorem \ref{thm:2}
by introducing fractional power control (FPC) with parameter $s
\in [0,1]$.  Under FPC the transmission power is set to $P_i =
\frac{p}{\mathbb{E}[H^{-s}]} H_{ii}^{-s}$ for each $i$. The
received power at receiver $i$ is then $P_i H_{ii} d^{-\alpha} =
\frac{p}{\mathbb{E}[H^{-s}]} H_{ii}^{1-s}d^{-\alpha}$, which
depends upon $i$ aside from $s=1$.  The expected transmission
power is $p$, ensuring a fair comparison with the results in
Theorems \ref{thm:1} and \ref{thm:2}.  Note that constant power
corresponds to $s=0$ and channel inversion corresponds to $s=1$.
The following theorem gives good approximations on the outage
probability and maximum allowable transmission intensity under
FPC.

\begin{theorem}\label{thm:3}
{\em {\bf Fractional power control.} Consider a network where the
SINR at the reference receiver is given by (\ref{eq:opdef}) with
$P_i = \frac{p}{\mathbb{E}[H^{-s}]} H_{ii}^{-s}$ for all $i$, for
some $s \in [0,1]$.  Then the following expressions give good
approximations of the outage probability and maximum transmission
attempt intensity for $\lambda, ~ \epsilon$ small
\begin{eqnarray*}
q^{\rm fpc}(\lambda) & \geq & q^{\rm fpc}_l(\lambda) = 1 -
\mathbb{P} \left(H_{00} \geq \kappa(s) \right) \times \\
&& \hspace{25mm}  \mathbb{E} \left[ \exp \left\{- \lambda  \pi d^2
\mathbb{E}[H^{-s \delta}] \mathbb{E}[H^{\delta}] \left(
\frac{H_{00}^{1-s} }{\beta } - \frac{\mathbb{E}[H^{-s}]}{\snr}
\right)^{-\delta} \right\} \Big\vert H_{00} \geq \kappa(s) \right] \\
    & \approx & \tilde{q}^{\rm fpc}_l(\lambda) = 1 - \mathbb{P} \left(H_{00} \geq \kappa(s) \right) \times \\
&& \hspace{25mm} \exp \left\{- \lambda \pi d^2  \mathbb{E}[H^{-s
\delta}] \mathbb{E}[H^{\delta}]
    \mathbb{E}  \left[ \left(\frac{H_{00}^{1-s} }{\beta} -  \frac{\mathbb{E}[H^{-s}]}{\snr} \right)^{-\delta} \Big\vert
H_{00} \geq \kappa(s)  \right] \right\}\\
    \lambda^{\rm fpc}(\epsilon) & \approx & \tilde{\lambda}^{\rm fpc}(\epsilon) =
    - \log \left( \frac{1-\epsilon}{\mathbb{P} \left(H_{00} \geq \kappa(s) \right)} \right) \frac{1}{\pi d^2} \frac{1} {\mathbb{E}[H^{-s\delta}]\mathbb{E}[H^{\delta}]}
\times \\
&& \hspace{65mm}    \left(  \mathbb{E}
    \left[ \left(\frac{H_{00}^{1-s} }{\beta } -  \frac{\mathbb{E}[H^{-s}]}{\snr} \right)^{-\delta} \Big\vert
H_{00} \geq \kappa(s)  \right] \right)^{-1}
\end{eqnarray*}}
where $\kappa(s) = \left( \frac{\beta}{\rm SNR} \mathbb{E}[H^{-s}]
\right)^{\frac{1}{1-s}}$.

\end{theorem}
\begin{proof}
Under FPC, the transmit power for each user is constructed as $P_i =
\frac{p}{\mathbb{E}[H^{-s}]}H_{ii}^{-s}$.  Substituting this value
into the proof for Theorem 2 immediately gives the expression for
$q^{\rm fpc}_l(\lambda)$. Again, the transmission attempt intensity
approximation is obtained by solving $\tilde{q}_l(\lambda) =
\epsilon$ for $\lambda$.
\end{proof}

As with Theorem \ref{thm:2}, the approximation $q^{\rm
fpc}_l(\lambda) \approx \tilde{q}^{\rm fpc}_l(\lambda)$ is accurate
when the exponential term in $q^{\rm fpc}_l(\lambda)$ is
approximately linear in its argument and thus Jensen's is tight.  In
other words, this approximation utilizes the fact that ${\rm
e}^{-x}$ is nearly linear for small $x$.  Looking at the expression
for $q^{\rm fpc}_l(\lambda)$ we see that this reasonable when the
\textit{relative density} $\lambda \pi d^2$ is small.  If this is
not true then the approximation $\tilde{q}^{\rm fpc}_l(\lambda)$ is
not sufficiently accurate, as will be further seen in the numerical
results presented in Section \ref{sec:numerical}.
 The FPC transmission attempt intensity approximation,
$\tilde{\lambda}^{\rm fpc}(\epsilon)$, is obtained by solving
$\tilde{q}^{\rm fpc}_l(\lambda) = \epsilon$ for $\lambda$. The
following corollary illustrates the simplification of the above
results when the noise may be ignored.

\begin{corollary}
\label{cor:3}
{\em When $\eta = 0$ the expressions in Theorem \ref{thm:3} simplify to:
\begin{eqnarray}
\label{eqn:c}
q^{\rm fpc}(\lambda) &\geq& q^{\rm fpc}_l(\lambda) = 1 - \mathbb{E} \left[ \exp \left\{-\lambda \pi d^2 \beta^{\delta} \mathbb{E}\left[H^{\delta} \right] \mathbb{E}\left[H^{-s\delta} \right] H_{00}^{-(1-s) \delta}\right\}\right] \nonumber \\
& \approx & \tilde{q}^{\rm fpc}_l(\lambda) = 1 - \exp \left\{-\lambda \pi d^2 \beta^{\delta} \mathbb{E}\left[H^{\delta} \right] \mathbb{E}\left[H^{-s\delta} \right] \mathbb{E} \left[ H^{-(1-s) \delta} \right] \right\}, \nonumber \\
\lambda^{\rm fpc}(\epsilon) & \approx & \tilde{\lambda}^{\rm fpc}(\epsilon) = - \log(1-\epsilon) \frac{1}{\pi d^2 \beta^{\delta}} \frac{1}{\mathbb{E}\left[H^{\delta} \right] \mathbb{E}\left[H^{-s\delta} \right] \mathbb{E} \left[ H^{-(1-s) \delta} \right]}.
\end{eqnarray} }
\end{corollary}

The loss factor for FPC, $L^{\rm fpc}$, is the reduction in the
transmission capacity approximation relative to the pure pathloss
case:
\begin{equation}
L^{\rm fpc}(s) = \frac{1}{\mathbb{E}\left[H^{\delta} \right] \mathbb{E}\left[H^{-s\delta} \right] \mathbb{E} \left[ H^{-(1-s) \delta} \right]}.
\end{equation}
Clearly, the loss factor $L^{\mathrm{fpc}}$ for FPC depends on the
design choice of the exponent $s$.

\subsection{Optimal Fractional Power Control Exponent}
\label{sec:optfpc}

Fractional power control represents a balance between the extremes
of no power control and channel inversion.  The mathematical effect
of fractional power control is to replace the
$\mathbb{E}[H^{-\delta}]$ term with $\mathbb{E}[H^{-s \delta}]
\mathbb{E}[H^{-(1-s)\delta}]$.  This is because the signal fading is
{\em softened} by the power control exponent $-s$ so that it results
in a leading term of $H^{-(1-s)}$ (rather than $H^{-1}$) in the
numerator of the SINR expression, and ultimately to the
$\mathbb{E}[H^{-(1-s)\delta} ]$ term.  The interference power is
also softened by the fractional power control and leads to the
$\mathbb{E}[H^{-s\delta} ]$ term.

The key question of course lies in determining the optimal power
control exponent.  Although it does not seem possible to derive an
analytical expression for the exponent that minimizes the general
expression for $q^{\rm fpc}_l(\lambda)$ given in Theorem
\ref{thm:3}, we can find the exponent that minimizes the outage
probability approximation in the case of no noise.
\begin{theorem}
\label{thm:4}
{\em In the absence of noise ($\eta = 0$), the
fractional power control outage probability approximation,
$\tilde{q}^{\rm fpc}_l(\lambda)$, is minimized for
$s=\frac{1}{2}$.  Hence, the fractional power control transmission
attempt intensity approximation, $\tilde{\lambda}^{\rm
fpc}(\epsilon)$ is also maximized for $s = \frac{1}{2}$.}
\end{theorem}

\begin{proof}
Because the outage probability/transmission density approximations
depend on the exponent $s$ only through the quantity $\mathbb{E}
\left[H^{-s \delta } \right] \mathbb{E}\left[H^{-(1-s)\delta}
\right]$, it is sufficient to show that $\mathbb{E} \left[H^{-s
\delta } \right] \mathbb{E}\left[H^{-(1-s)\delta} \right] $ is
minimized at $s=\frac{1}{2}$.  To do this, we use the following
general result, which we prove in the Appendix.  For any
non-negative random variable $X$, the function
\begin{equation}
h(s) = \mathbb{E}\left[ X^{-s} \right] \mathbb{E} \left[ X^{s-1} \right],
\end{equation}
is convex in $s$ for $s \in \mathbb{R}$ with a unique minimum at $s
= \frac{1}{2}$.  Applying this result to random variable
$X=H^{\delta}$ gives the desired result.
\end{proof}

The theorem shows that transmission density is maximized, or
equivalently, outage probability is minimized, by balancing the
positive and negative effects of power control, which are reduction
of signal fading and increasing interference, respectively. Using an
exponent greater than $\frac{1}{2}$ {\emph over-compensates} for
signal fading and leads to interference levels that are too high,
while using an exponent smaller than $\frac{1}{2}$ leads to small
interference levels but an {\emph under-compensation} for signal
fading. Note that because the key expression $\mathbb{E} \left[
H^{-s \delta } \right] \mathbb{E} \left[ H^{-(1-s)\delta} \right]$
is convex, the loss relative to using $s=\frac{1}{2}$ increases
monotonically both as $s \to 0$ and $s \to 1$.

One can certainly envision ``fractional" power control schemes
that go even further.  For example, $s > 1$ corresponds to
``super" channel inversion, in which bad channels take resources
from good channels even more so than in normal channel inversion.
Not surprisingly, this is not a wise policy. Less obviously, $s <
0$ corresponds to what is sometimes called ``greedy" optimization,
in which good channels are given more resources at the further
expense of poor channels.  Waterfilling is an example of a greedy
optimization procedure.  But, since $\mathbb{E} \left[ H^{-s
\delta } \right] \mathbb{E} \left[ H^{-(1-s)\delta} \right]$
monotonically increases as $s$ decreases, it is clear that greedy
power allocations of any type are worse than even constant
transmit power under the SINR-target set up.

The numerical results in the next section show that FPC is very
beneficial relative to constant transmit power or channel
inversion. However, fading has a deleterious effect relative to no
fading even if the optimal exponent is used. To see this, note
that $x^{-\frac{1}{2}}$ is a convex function and therefore
Jensen's yields $\mathbb{E}[X^{-\frac{1}{2}}] \geq
(\mathbb{E}[X])^{-\frac{1}{2}}$ for any non-negative random
variable $X$. Applying this to $X=H^{\delta}$ we get $\left(
\mathbb{E} \left[ H^{-\frac{\delta}{2}} \right] \right)^2 \geq
\left( \mathbb{E}[H^{\delta}] \right)^{-1}$, which implies
\begin{eqnarray*}
L^{\rm fpc}(1/2) = \frac{1} { \mathbb{E} \left[ H^{\delta} \right] \left( \mathbb{E} \left[ H^{- \frac{\delta}{2} } \right] \right)^2 } \leq 1.
\end{eqnarray*}
Therefore, fractional PC cannot fully overcome fading, but it is
definitely a better power control policy than constant power
transmission or traditional power control (channel inversion).

\section{\label{sec:numerical} Numerical Results and Discussion}

In this section, the implications of fractional power control are
illustrated through numerical plots and analytical discussion. The
tightness of the bounds will be considered as a function of the
system parameters, and the choice of a robust FPC exponent $s$
will be proposed.  As default parameters, the simulations assume
\begin{equation}
\begin{array}{cccccc}
\alpha = 3, & \beta = 1 ~(0 ~{\rm dB}), & d = 10 {\rm m}, & {\rm
SNR} = \frac{p d^{-\alpha}}{\eta  } = 100 ~(20~ {\rm dB}), & \lambda
= 0.0001 ~ \frac{\rm users}{{\rm m}^2}.
\end{array}
\end{equation}
Furthermore, Rayleigh fading is assumed for the numerical results.

\subsection{Effect of Fading}

The benefit of fractional power control can be quickly illustrated
in Rayleigh fading, in which case the channel power $H$ is
exponentially distributed and the moment generating function is
therefore
\begin{equation}
\label{eqn:at} \mathbb{E}[H^t] = \Gamma(1+t),
\end{equation}
where $\Gamma(\cdot)$ is the standard gamma function. If
fractional power control is used, the transmission capacity loss
due to fading is
\begin{eqnarray} \label{eq-rayleigh}
L^{\rm fpc} = \frac{1} { \mathbb{E} \left[ H^{\delta} \right]
\mathbb{E} \left[ H^{-s \delta } \right] \mathbb{E}\left[
H^{(1-s)\delta} \right] } = \frac{1}{ \Gamma(1 + \delta) \cdot
\Gamma(1 - s \delta) \cdot \Gamma(1 - (1-s) \delta)}
\end{eqnarray}
In Fig. \ref{fig-pc1} this loss factor ($L$) is plotted as a
function of $s$ for path loss exponents $\alpha = \{2.1, 3, 4\}$.
Notice that for each value of $\alpha$ the maximum takes place at
$s=\frac{1}{2}$, and that the cost of not using fractional power
control is highest for small path loss exponents because
$\Gamma(1+x)$ goes to infinity quite steeply as $x \rightarrow
-1$.  This plot implies that in severe fading channels, the gain
from FPC can be quite significant.

It should be noted that the expression in (\ref{eq-rayleigh}) is for
the case of no thermal noise ($\eta=0$).  In this case the power
cost of FPC completely vanishes, because the same power
normalization (by $\mathbb{E}[H^{-s}]$) is performed by each
transmitting node and therefore this normalization cancels in the
SIR expression.  On the other hand, this power cost does not vanish
if the noise is strictly positive and can potentially be quite
significant, particularly if $\snr$ is not large.  A simple
application of Jensen's shows that the power normalization factor
$\mathbb{E}[H^{-s}]$ is an increasing function of the exponent $s$
for any distribution on $H$.  For the particular case of Rayleigh
fading this normalization factor is $ \Gamma(1-s)$ which makes it
prohibitively expensive to choose $s$ very close to one; indeed, the
choice $s=1$ requires infinite power and thus is not feasible.  On
the other hand, note that $\Gamma(.5)$ is approximately $2.5$ dB and
thus the cost of a moderate exponent is not so large.  When the
interference-free $\snr$ is reasonably large, this normalization
factor is relatively negligible and the effect of FPC is well
approximated by (\ref{eq-rayleigh}).

\subsection{Tightness of Bounds}

 There are two principle approximations made in attaining the
expressions for outage probability and transmission capacity in
Theorem 3.  First, the inequality is due to considering only
\emph{dominant} interferers; that is, an interferer whose channel to
the desired receiver is strong enough to cause outage even without
any other interferers present.  This is a lower bound on outage
since it ignores non-dominant interferers, but nevertheless has been
seen to be quite accurate in our prior work
\cite{WebYan05,WebAnd07,WebAndJin07}.  Second, Jensen's inequality
is used to bound $\mathbb{E}[\exp(X)] \geq \exp(\mathbb{E}[X])$ in
the opposite direction, so this results in an approximation to the
outage probability rather than a lower bound; numerical results
confirm that this approximation is in fact not a lower bound in
general. Therefore, we consider the three relevant quantities: (1)
the actual outage probability $q^{\rm fpc}(\lambda)$, which is
determined via Monte-Carlo simulation and does not depend on any
bounds or approximations, (2) a numerical computation of the outage
probability lower bound $q^{\rm fpc}_l(\lambda)$, and (3) the
approximation to the outage probability $\tilde{q}^{\rm
fpc}_l(\lambda)$ reached by applying Jensen's inequality to $q^{\rm
fpc}_l(\lambda)$. Note that because of the two opposing bounds (one
lower and one upper), we cannot say \emph{a priori} that method (2)
will produce more accurate expressions than method (3).

The tightness of the bounds is explored in Figs. \ref{fig:default} -
\ref{fig:SINR}. Consider first Fig. \ref{fig:default} for the
default parameters given above. We can see that the lower bound and
the Jensen approximation both reasonably approximate the simulation
results, and the approximation winds up serving as a lower bound as
well. The Jensen's approximation is very accurate for large values
of $s$ (i.e., closer to channel inversion), and while looser for
smaller values of $s$, this ``error" actually moves the Jensen's
approximation closer to the actual (simulated) outage probability.
The Jensen's approximation approaches the lower bound as $s \to 1$
because the random variable $H^{(1-s)\delta}$ approaches a constant,
where Jensen's inequality trivially holds with equality (see, e.g.,
(\ref{eqn:c})). Changing the path loss exponent $\alpha$, the SNR,
the target SINR $\beta$, or the density $\lambda$ can have a
significant effect on the bounds, as we will see.  With the
important exception of high density networks, the approximations are
seen to be reasonably accurate for reasonable parameter values.

\textbf{Path loss.} In Fig. \ref{fig:PL}, the bounds are given for
$\alpha = 2.2$ and $\alpha = 5$, which correspond to much weaker and
much stronger attenuation than the (more likely) default case of
$\alpha = 3$. For weaker attenuation, we can see that the lower
bound holds the right shape but is less accurate, while the Jensen's
approximation becomes very loose when the FPC exponent $s$ is small.
For path loss exponents near $2$, the dominant interferer
approximation is weakened because the attenuation of non-dominant
interferers is less drastic. On the other hand, both the lower bound
and Jensen's approximation are very accurate in strong attenuation
environments as seen in the $\alpha=5$ plot.  This is because the
dominant interferer approximation is very reasonable in such cases.

\textbf{SNR.}  The behavior of the bounds also varies as the
background noise level changes, as shown in Fig. \ref{fig:SNR}. When
the SNR is 10 dB, the bounds are quite tight.  However, the behavior
of outage probability as a function of $s$ is quite different from
the default case in Fig. \ref{fig:default}: outage probability
decreases slowly as $s$ is increased, and a rather sharp jump is
seen as $s$ approaches one.  When the interference-free SNR is only
moderately larger than the target SINR (in this case there is a 10
dB difference between $\snr$ and $\beta$), a significant portion of
outages occur because the signal power is so small that the
\textit{interference-free} received SNR falls below the target
$\beta$; this probability is captured by the $\mathbb{P}
\left(H_{00} \geq \kappa(s) \right)$ terms in Theorem \ref{thm:3}.
On the other hand, if $\snr$ is much larger than the target $\beta$,
outages are almost always due to a combination of  signal fading and
large interference power rather than to signal fading alone (i.e.,
$\mathbb{P} \left(H_{00} \geq \kappa(s) \right)$ is insignificant
compared to the total outage probability).  When outages caused
purely by signal fading are significant, the dependence on the
exponent $s$ is significantly reduced.  Furthermore, the power cost
of FPC becomes much more significant when the gap between $\snr$ and
$\beta$ is reduced; this explains the sharp increase in outage as
$s$ approaches one. When $\snr = 30$ dB, the behavior is quite
similar to the 20 dB case because at this point the gap between
$\snr$ and $\beta$ is so large that thermal noise can effectively be
neglected.

\textbf{Target SINR.}  A default SINR of $\beta = 1$ was chosen,
which corresponds roughly to a spectral efficiency of 1 bps/Hz with
strong coding, and lies between the low and high SINR regimes.
Exploring an order of magnitude above and below the default in Fig.
\ref{fig:SINR}, we see that for $\beta = 0.1$ the bounds are highly
accurate, and show that $s^* = \frac{1}{2}$ is a good choice.  For
this choice of parameters there is a 30 dB gap between $\snr$ and
$\beta$ and thus thermal noise is essentially negligible.  On the
other hand, if $\beta = 10$ the bounds are still reasonable, but the
outage behavior is very similar to the earlier case where $\snr=10$
dB and $\beta=0$ dB because there is again only a 10 dB gap between
$\snr$ and $\beta$. Despite the qualitative and quantitative
differences for low SNR and high target SINR from the default
values, it is interesting to note that in both cases $s =
\frac{1}{2}$ is still a robust choice for the FPC exponent.

\textbf{Density.}  The default value of $\lambda = 0.0001$
corresponds to a somewhat low density network because the expected
distance to the nearest interferer is approximately $50$ m, while
the TX-RX distance is $d=10$ m.  In Fig. \ref{fig:DENSITY} we
explore a density an order of magnitude lower and higher than the
default value.  When the network is even sparser, the bounds are
extremely accurate and we see that $s^* = \frac{1}{2}$ is a
near-optimal choice.  However, the behavior with $s$ is very
different in a dense network where $\lambda = .001$ and the nearest
interferer is approximately $17$ m away.  In such a network we see
that the nearest neighbor bound is quite loose because a substantial
fraction of outages are caused by the summation of non-dominant
interferers, as intuitively expected for a dense network.  Although
the bound is loose, it does capture the fact that outage increases
with the exponent $s$. On the other hand, the Jensen approximation
is loose and does not correctly capture the relationship between $s$
and outage.  The approximation is based on the fact that the
function $e^{-x}$ is approximately linear for small $x$. The
quantity $x$ is proportional to $\pi \lambda d^2$, which is large
when the network is dense relative to TX-RX distance $d$, and thus
this approximation is not valid for relatively dense networks.

\subsection{Choosing the FPC exponent $s$}

Determining the optimum choice of FPC exponent $s$ is a key interest
of this paper.  As seen in Sect. \ref{sec:optfpc}, $s^* =
\frac{1}{2}$ is optimal for the Jensen's approximation and with no
noise, both of which are questionable assumptions in many regimes of
interest.  In Figs. \ref{fig:sPL} -- \ref{fig:sDENSITY}, we plot the
truly optimal choice of $s^*$ for the default parameters, while
varying $\alpha$, SNR, $\beta$, and $\lambda$, respectively.  That
is, the value of $s$ that minimizes the true outage probability is
determined for each set of parameters.  The FPC exponents
$s_l(\Delta)$ and $s_u(\Delta)$ are also plotted, which provide
$\Delta$\% error below and above the optimum outage probability. For
the plots, we let $\Delta = 1$ and $\Delta = 10$.

The key findings are: (1) In the pathloss ($\alpha$) plot, $s^* =
\frac{1}{2}$ is a very robust choice for all attenuation regimes;
(2) For SNR, $s^* = \frac{1}{2}$ is only robust at high SNR, and at
low SNR constant transmit power is preferable; (3) For target SINR
$\beta$, $s^* = \frac{1}{2}$ is robust at low and moderate SINR
targets (i.e. low to moderate data rates), but for high SINR targets
constant transmit power is preferred; (4) For density $\lambda$,
$s^* = \frac{1}{2}$ is robust at low densities, but constant
transmit power is preferred at high densities.

 The explanation
for findings (2) and (3) is due to the dependence of outage behavior
on the difference between $\snr$ and $\beta$.  As seen earlier,
thermal noise is essentially negligible when this gap is larger than
approximately 20 dB.  As a result, it is reasonable that the
exponent shown to be optimal for noise-free networks
($s=\frac{1}{2}$) would be near-optimal for networks with very low
levels of thermal noise.  On the other hand, outage probability
behaves quite differently when $\snr$ is only slightly larger than
$\beta$.  In this case, power is very valuable and it is not worth
incurring the normalization cost of FPC and thus very small FPC
exponents are optimal. Intuitively, achieving high data rates in
moderate SNR or moderate data rates in low SNR are difficult
objectives in a decentralized network.  The low SNR case is somewhat
anomalous, since the SNR is close to the target SINR, so almost no
interference can be tolerated.  Similarly, to meet a high SINR
constraint in a random network of reasonable density, the outage
probability must be quite high, so this too may not be particularly
meaningful.

To explain (4), recall that the Jensen-based approximation to outage
probability is not accurate for dense networks and the plot shows
that constant power ($s=0$) is preferred at high
densities.\footnote{Based on the figure it may appear that choosing
$s<0$, which means users with good channels transmit with additional
power, outperforms constant power transmission.  However, numerical
results (not shown here) indicate that this provides a benefit only
at extremely high densities for which outage probability is
unreasonably large.  Intuitively, a user with a poor channel in a
dense network is extremely unlikely to be able to successfully
communicate and global performance is improved by having such a user
not even attempt to transmit, as done in the threshold-based policy
studied in \cite{WebAndJin07}.} Fractional power control softens
signal fading at the expense of more harmful interference power, and
this turns out to be a good tradeoff in relatively sparse networks.
In dense networks, however, there generally are a large number of
nearby interferers and as a result the benefit of reducing the
effect of signal fading (by increasing exponent $s$) is overwhelmed
by the cost of more harmful interference power.  Note that this is
consistent with results on channel inversion ($s=1$) in
\cite{WebAndJin07}, where $s=0$ and $s=1$ are seen to be essentially
equivalent at low densities (as expected by the Jensen
approximation) but inversion is inferior at high densities.

\section{\label{sec:5} Possible Areas for Future Study}

Given the historically very high level of interest in the subject
of power control for wireless systems, this new approach for power
control opens many new questions. It appears that FPC has
potential for many applications due to its inherent simplicity,
requirement for only simple pairwise feedback, and possible
\emph{a priori} design of the FPC parameter $s$. Some areas that
we recommend for future study include the following.

\textbf{How does FPC perform in cellular systems?}.  Cellular
systems in this case are harder to analyze than ad hoc networks,
because the base stations (receivers) are located on a regular
grid and thus the tractability of the spatial Poisson model cannot
be exploited.  On the other hand, FPC may be even more helpful in
centralized systems.  Note that some numerical results for
cellular systems are given in reference \cite{XiaRat06}, but no
analysis is provided.

\textbf{Can FPC be optimized for spectral efficiency?}.  In this
paper we have focused on outage relative to an SINR constraint as
being the metric.  Other metrics can be considered, for example
maximizing the average spectral efficiency, i.e. $\max
\mathbb{E}[\log_2(1+{\rm SINR})]$, which could potentially result in
optimal exponents $s < 0$, which is conceptually similar to
waterfilling.

\textbf{What is the effect of scheduling on FPC?}  If scheduling
is used, then how should power levels between a transmitter and
receiver be set?  Will $s = \frac{1}{2}$ still be optimal?  Will
the gain be increased or reduced?  We conjecture that the gain
from FPC will be smaller but non-zero for most any sensible
scheduling policy, as the effect of interference inversion is
softened.

\textbf{Can FPC be used to improve iterative power control?}  At
each step of the Foschini-Miljanic algorithm (as well as most of its
variants), transmitters adjust their power in a manner similar to
channel-inversion, i.e., each transmitter fully compensates for the
current SINR.  While this works well when the target SINR's are
feasible, it does not necessarily work well when it is not possible
to satisfy all users' SINR requirements.  In such a setting, it may
be preferable to perform \emph{partial} compensation for the current
SINR level during each iteration.  For example, if a link with a 10
dB target is currently experiencing an SINR of 0 dB, rather than
increasing its transmit power by 10 dB to fully compensate for this
gap (as in the Foschini-Miljanic algorithm), an FPC-motivated
iterative policy might only boost power by 5 dB (e.g., adjust power
in linear units according to the square root of the gap).

\section{\label{sec:6} Conclusions}

This paper has applied fractional power control as a general
approach to pairwise power control in decentralized (e.g. ad hoc or
spectrum sharing) networks. Using two approximations, we have shown
that a fractional power control exponent of $s^* = \frac{1}{2}$ is
optimal in terms of outage probability and transmission capacity, in
contrast to constant transmit power ($s=0$) or channel inversion
($s=1$) in networks with a relatively low density of transmitters
and low noise levels.  This implies that there is an optimal balance
between compensating for fades in the desired signal and amplifying
interference.  We saw that a gain on the order of $50\%$ or larger
(relative to no power control or channel inversion) might be typical
for fractional power control in a typical wireless channel.

\appendix

We prove that for any non-negative random variable $X$, the function
\begin{equation}
h(s) = \mathbb{E}\left[ X^{-s} \right] \mathbb{E} \left[ X^{s-1}
\right],
\end{equation}
is convex in $s$ for $s \in \mathbb{R}$ with a unique minimum at $s
= \frac{1}{2}$.  In order to show $h(s)$ is convex, we show $h$ is
log-convex and use the fact that a log-convex function is convex. We
define
\begin{equation}
H(s) = \log h(s) = \log \left( \mathbb{E}\left[X^{-s} \right]
\mathbb{E}\left[X^{s-1} \right] \right),
\end{equation}
and recall H\"{o}lder's inequality:
\begin{equation}
\mathbb{E}[XY] \leq \left( \mathbb{E}[X^p]
\right)^{\frac{1}{p}}\left( \mathbb{E}[Y^q] \right)^{\frac{1}{q}},
~~~~ \frac{1}{p}+\frac{1}{q}=1.
\end{equation}
The function $H(s)$ is convex if $H(\lambda s_1 + (1-\lambda)s_2)
\leq \lambda H(s_1) + (1-\lambda) H(s_2)$ for all $s_1,s_2$ and all
$\lambda \in [0,1]$. Using H\"{o}lder's with $p = \frac{1}{\lambda}$
and $q=\frac{1}{1-\lambda}$ we have:
\begin{eqnarray}
H(\lambda s_1 + (1-\lambda)s_2) &=& \log \left( \mathbb{E}\left[X^{-(\lambda s_1 + (1-\lambda)s_2)} \right] \mathbb{E}\left[X^{(\lambda s_1 + (1-\lambda)s_2)-1} \right] \right) \nonumber \\
&=& \log \left( \mathbb{E}\left[X^{-\lambda s_1} X^{(1-\lambda)s_2} \right] \mathbb{E}\left[X^{\lambda (s_1-1)} X^{(1-\lambda)(s_2-1)} \right] \right) \nonumber \\
& \leq & \log \left( \mathbb{E}\left[X^{-s_1} \right]^{\lambda} \mathbb{E} \left[ X^{s_2} \right]^{1-\lambda} \mathbb{E}\left[X^{s_1-1}\right]^{\lambda} \mathbb{E} \left[X^{s_2-1} \right]^{1-\lambda} \right) \nonumber \\
& = & \lambda \log \left( \mathbb{E}\left[X^{-s_1} \right] \mathbb{E}\left[X^{s_1-1}\right] \right) + (1-\lambda) \log \left( \mathbb{E} \left[ X^{s_2} \right] \mathbb{E} \left[X^{s_2-1} \right] \right) \nonumber \\
&=& \lambda H(s_1) + (1-\lambda)  H(s_2).
\end{eqnarray}
This implies $H(s)$ is convex, which further implies convexity of
$h(s)$. The derivative of $h$ is
\begin{equation}
h'(s) = \mathbb{E}\left[ X^{-s} \right] \mathbb{E}\left[ X^{s-1}
\log X \right] - \mathbb{E}\left[ X^{s-1} \right] \mathbb{E}\left[
X^{-s} \log X \right],
\end{equation}
and it can easily be seen that $s^* = \frac{1}{2}$ is the unique
minimizer satisfying $h'(s) = 0$.


\pagebreak

\begin{figure}
\centering
\includegraphics[width=3.5in]{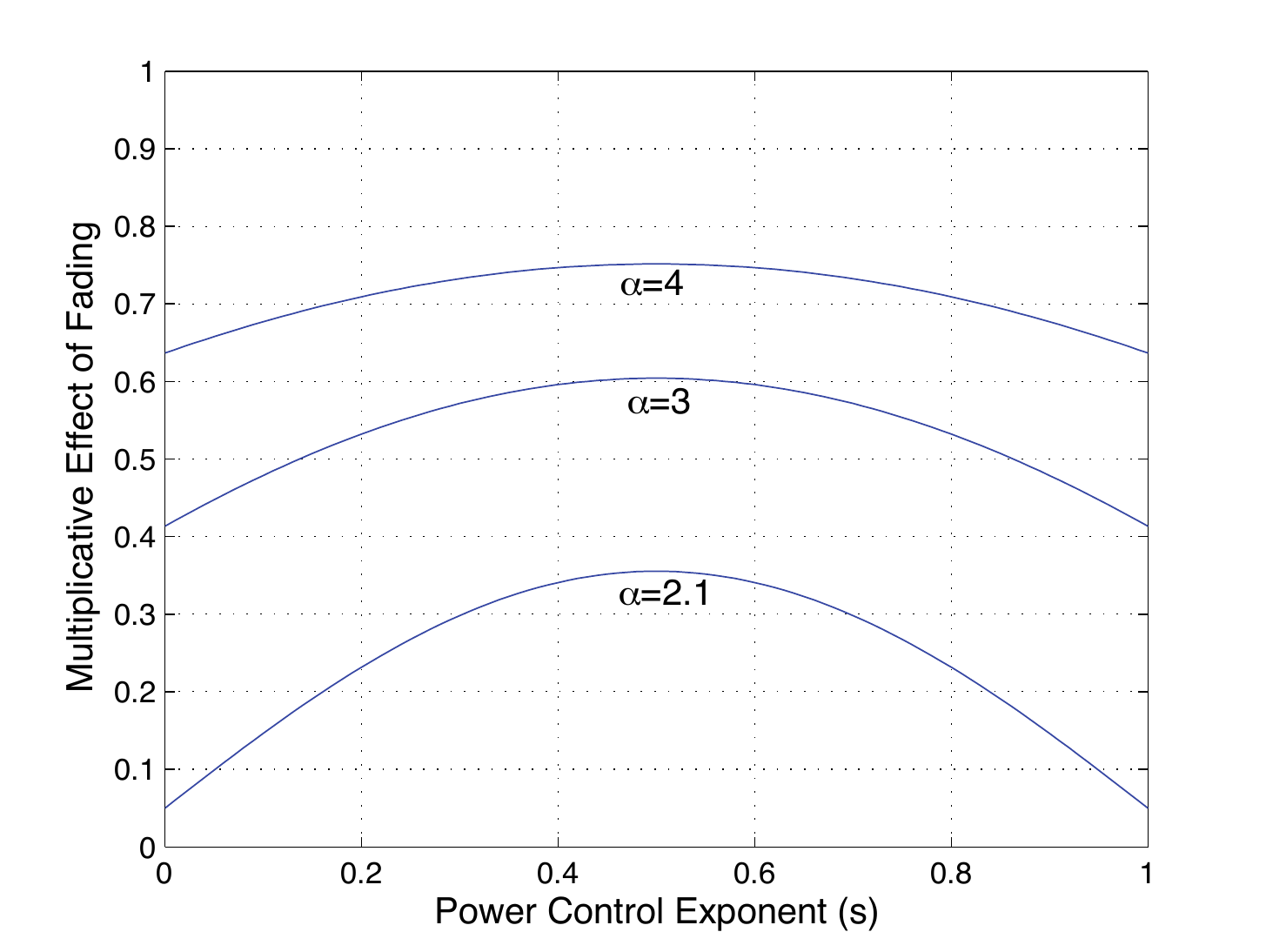}
\caption{The loss factor $L$ vs. $s$ for Rayleigh fading.  Note
that $L^{\mathrm{cp}}$ and $L^{\mathrm{ci}}$ are the left edge and
right edge of the plot, respectively.} \label{fig-pc1}
\end{figure}

\begin{figure}
\centering
\includegraphics[width=3.5in]{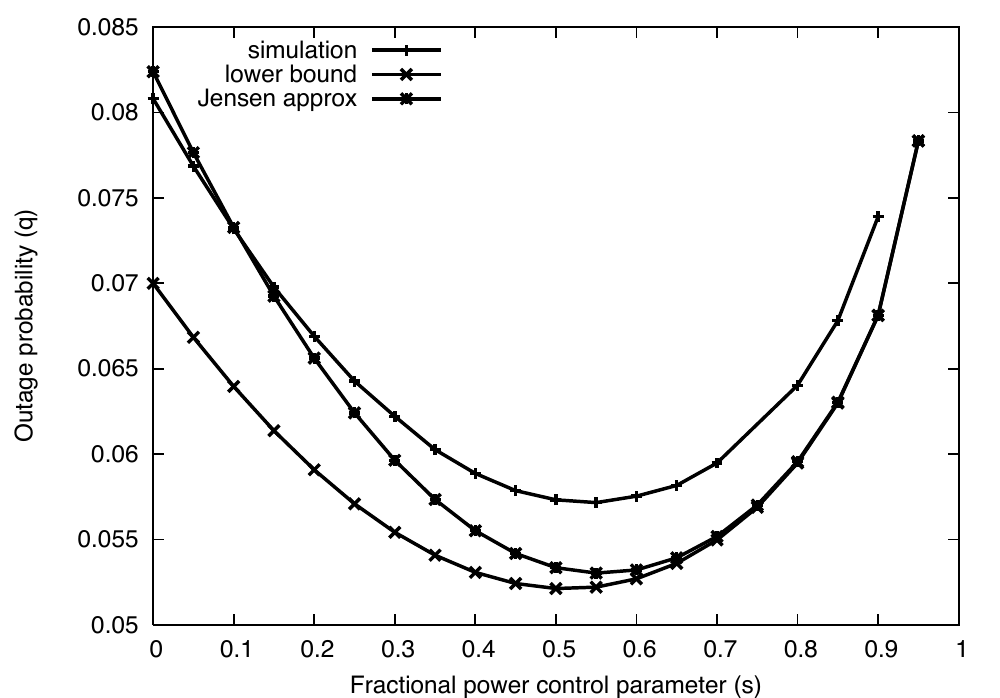}
\caption{The outage probability (simulated, lower bound, and
Jensen's approximation) vs. FPC exponent $s$ for the default
parameters.} \label{fig:default}
\end{figure}

\begin{figure}
\centering
\includegraphics[width=3.5in]{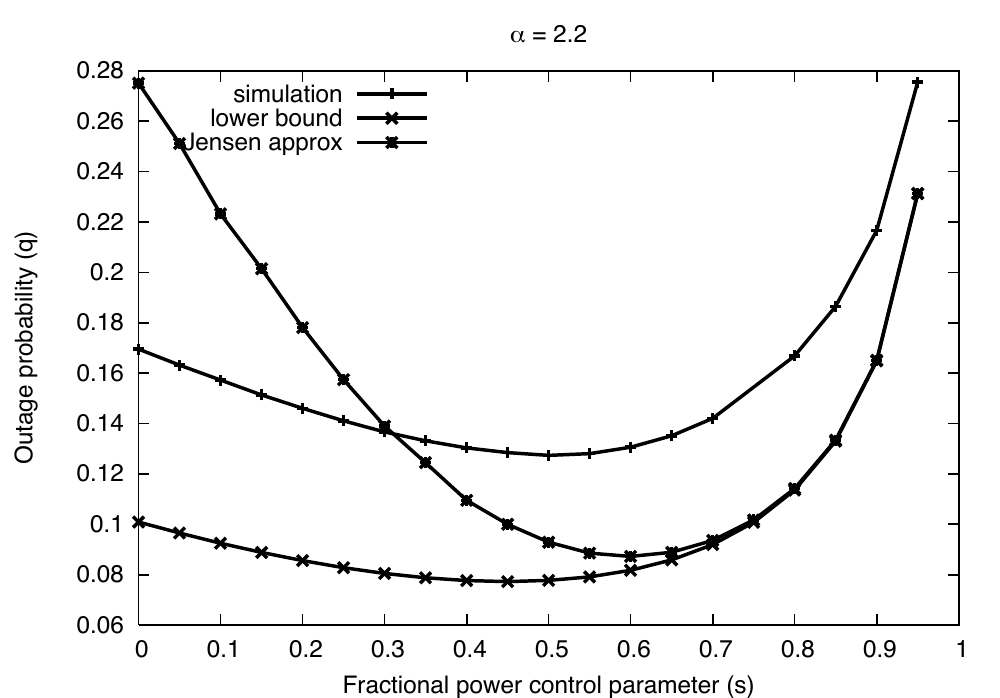}
\includegraphics[width=3.5in]{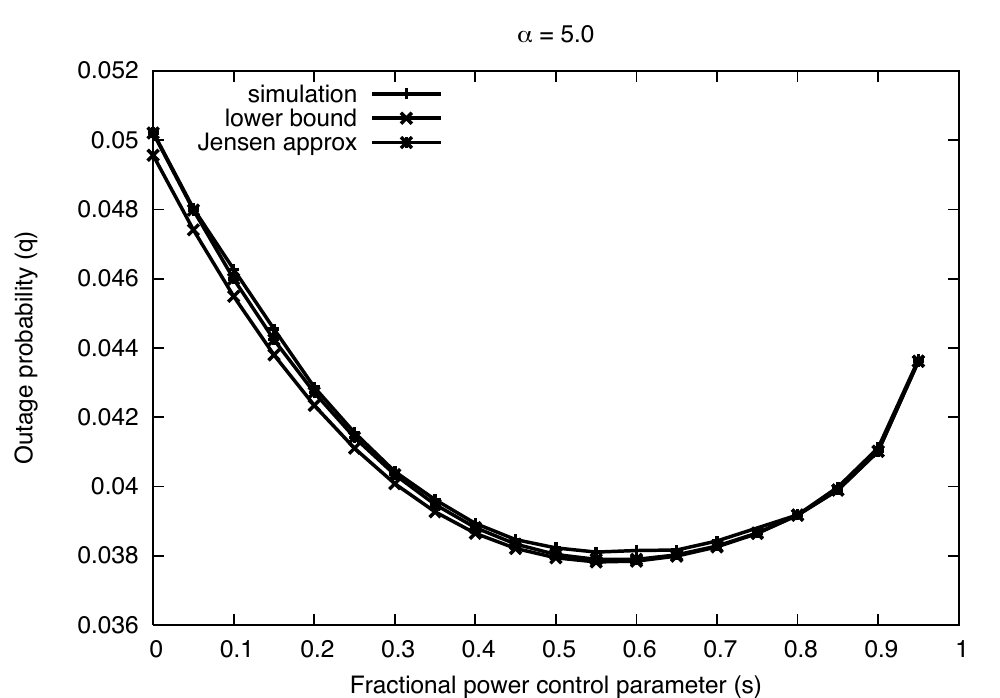}
\caption{The outage probability (simulated, lower bound, and
Jensen's approximation) vs. FPC exponent $s$ for $\alpha = 2.2$
(left) and $\alpha = 5$ (right).} \label{fig:PL}
\end{figure}

\begin{figure}
\centering
\includegraphics[width=3.5in]{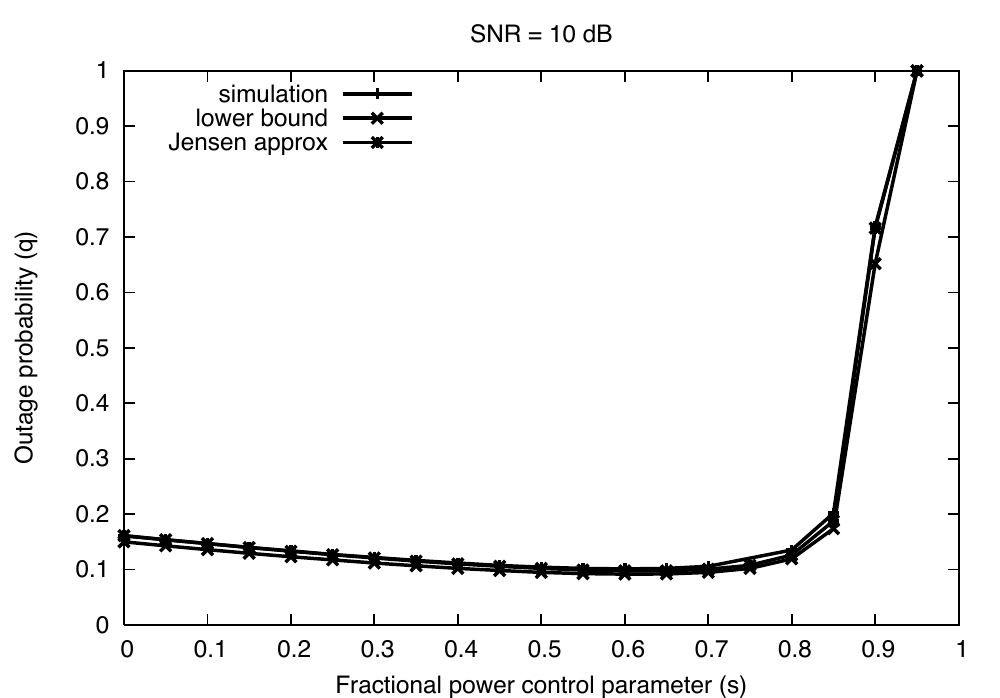}
\includegraphics[width=3.5in]{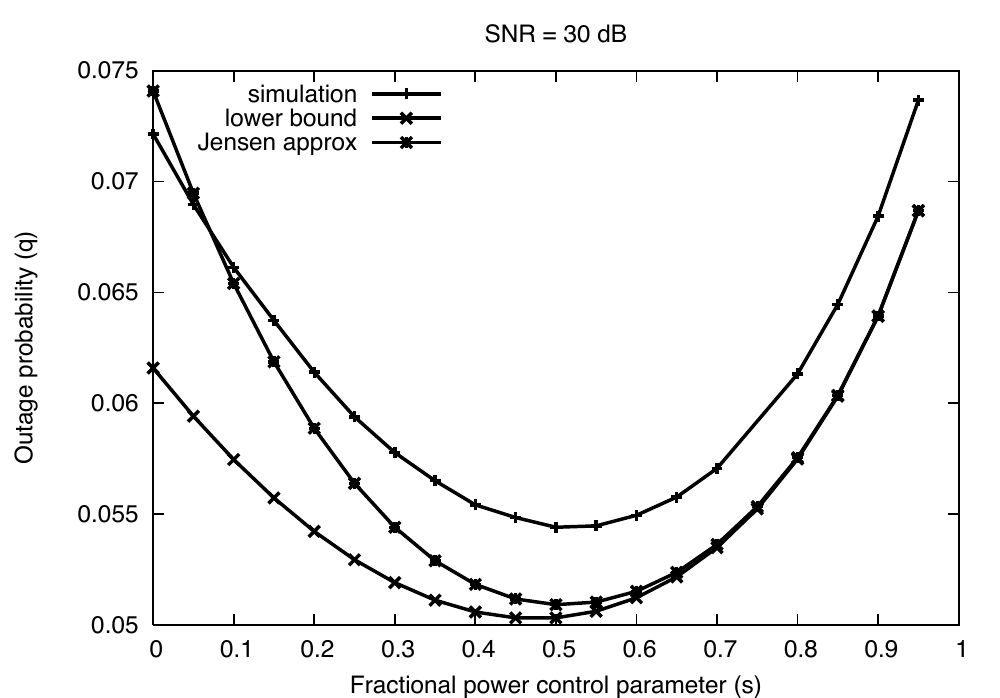}
\caption{The outage probability (simulated, lower bound, and
Jensen's approximation) vs. FPC exponent $s$ for ${\rm SNR} = 10$ dB
(left) and ${\rm SNR} = 30$ dB (right).} \label{fig:SNR}
\end{figure}

\begin{figure}
\centering
\includegraphics[width=3.5in]{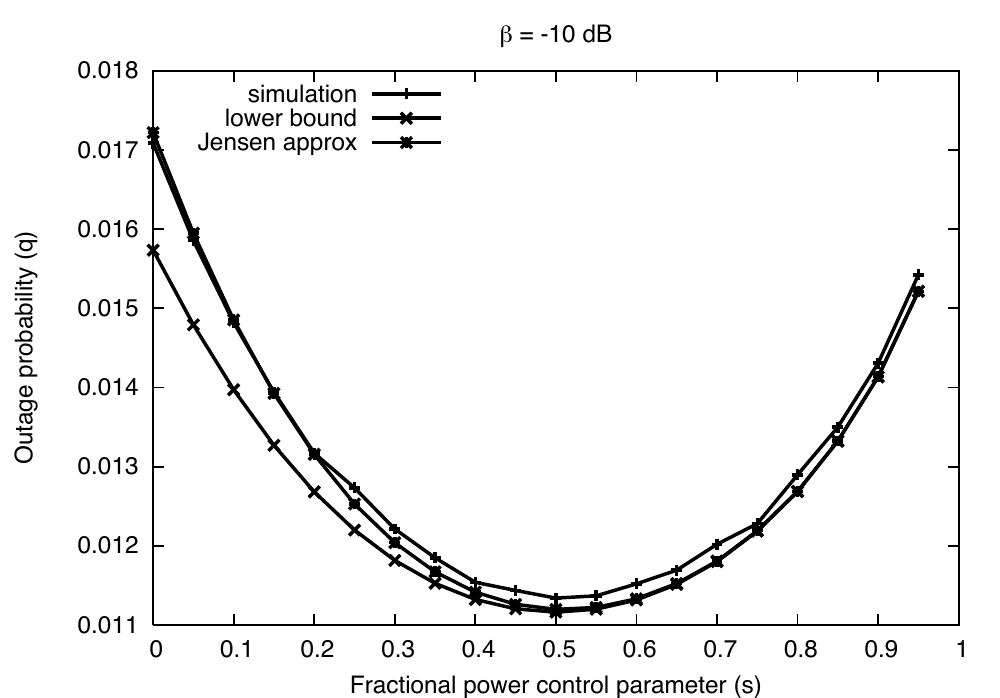}
\includegraphics[width=3.5in]{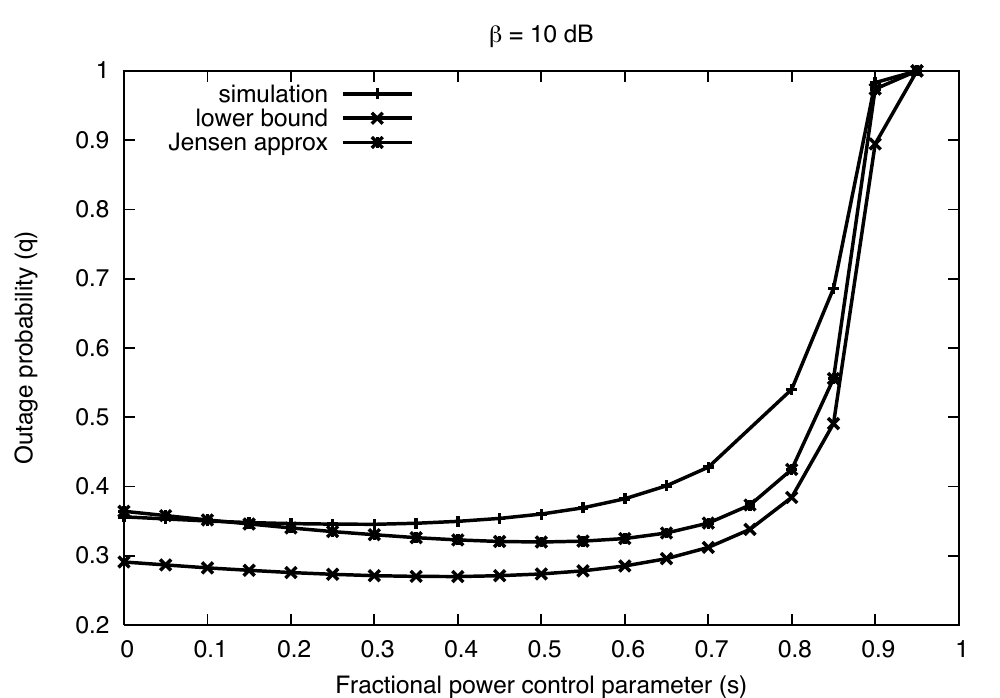}
\caption{The outage probability (simulated, lower bound, and
Jensen's approximation) vs. FPC exponent $s$ for $\beta = -10$ dB
(left) and $\beta = 10$ dB (right).} \label{fig:SINR}
\end{figure}

\begin{figure}
\centering
\includegraphics[width=3.5in]{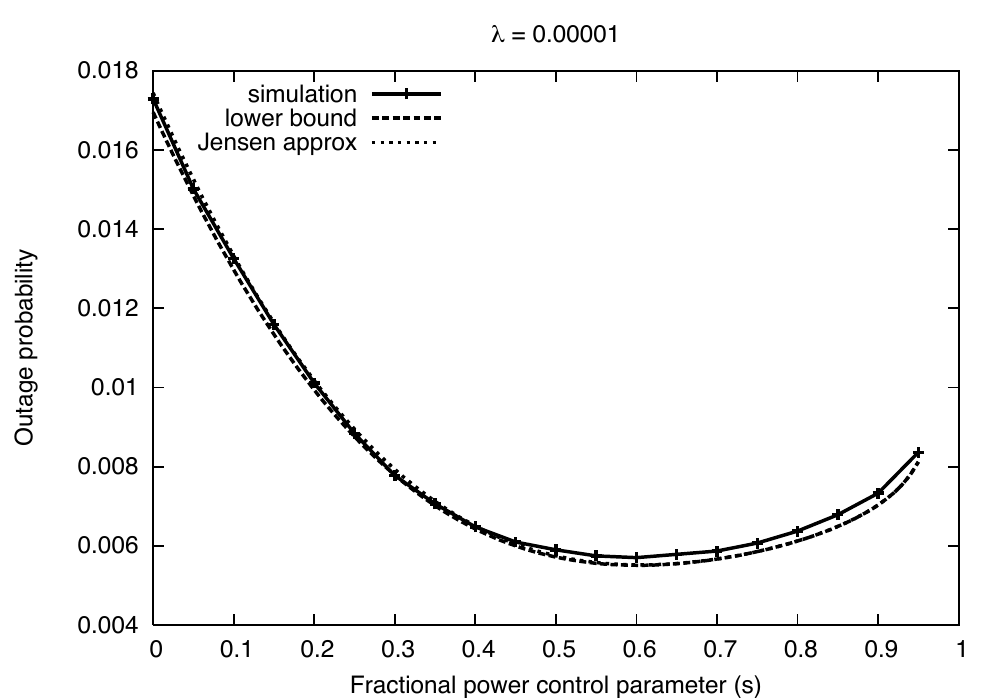}
\includegraphics[width=3.5in]{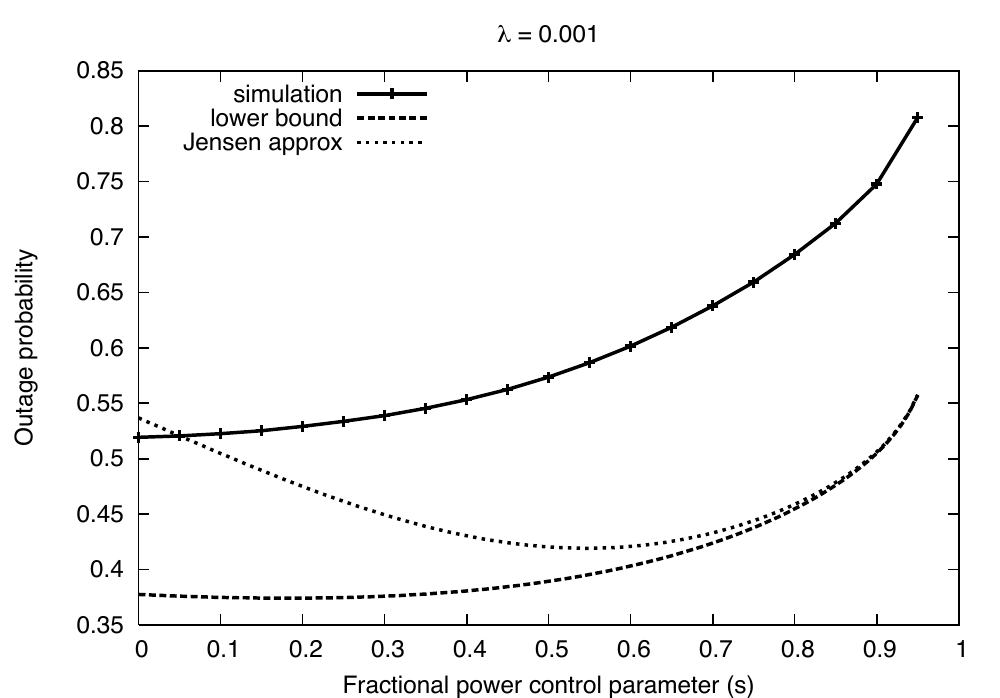}
\caption{The outage probability (simulated, lower bound, and
Jensen's approximation) vs. FPC exponent $s$ for $\lambda = 0.00001$
(left) and $\lambda = 0.001$ (right).} \label{fig:DENSITY}
\end{figure}

\begin{figure}
\centering
\includegraphics[width=3.5in]{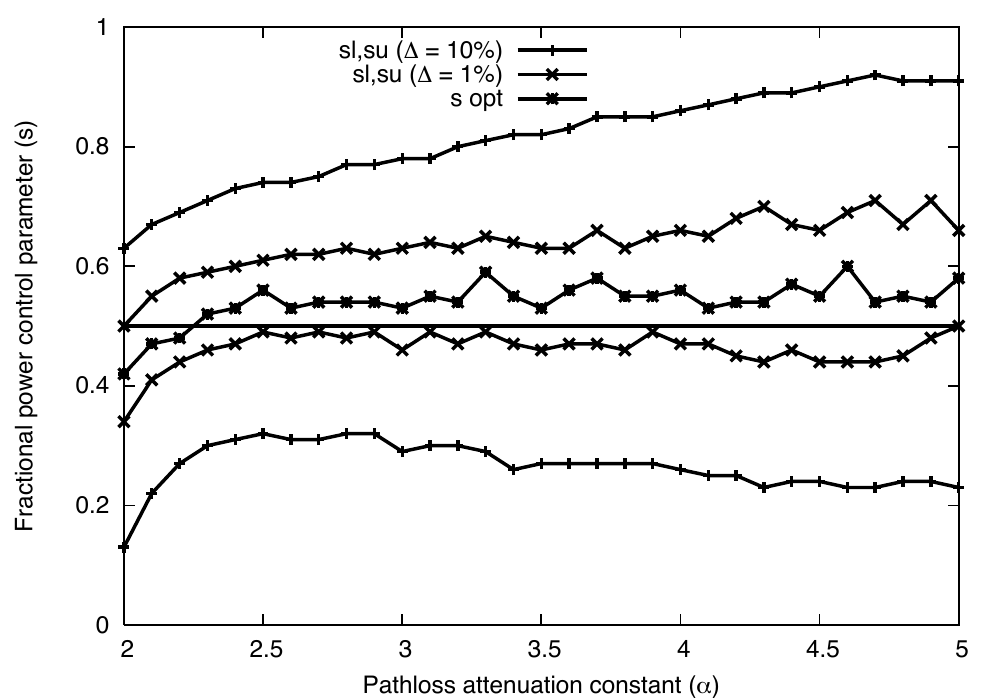}
\caption{The optimal choice of FPC exponent $s$ vs. PL exponent
$\alpha$, with $\pm 1$\% and $\pm10$\% selections for $s$.}
\label{fig:sPL}
\end{figure}

\begin{figure}
\centering
\includegraphics[width=3.5in]{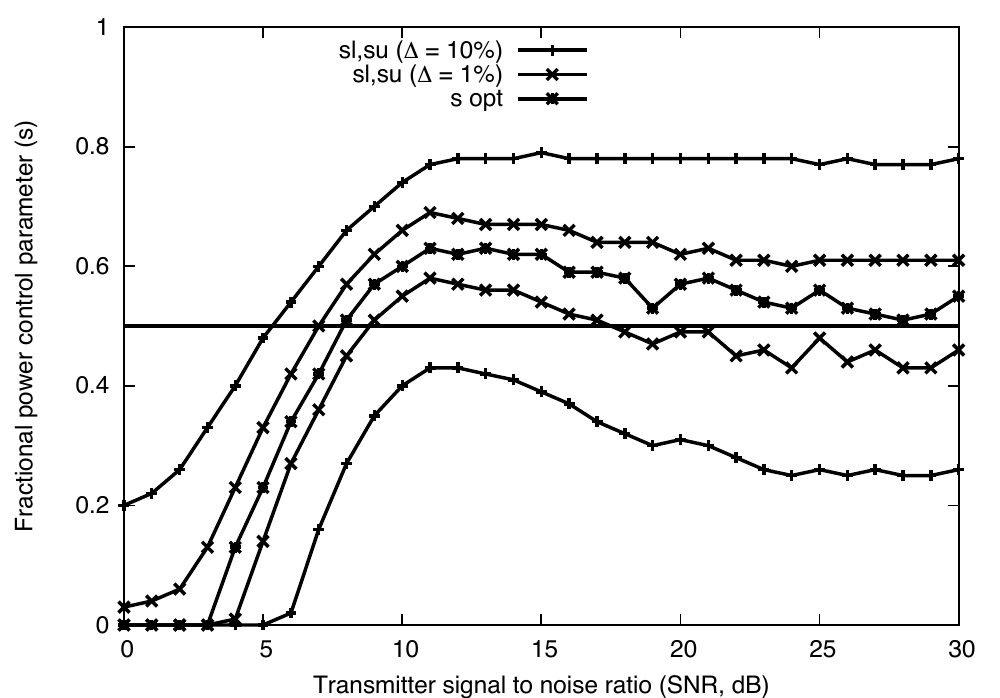}
\caption{The optimal choice of FPC exponent $s$ vs. transmitter SNR
$= \frac{\rho}{\eta}$, with $\pm 1$\% and $\pm10$\% selections for
$s$.} \label{fig:sSNR}
\end{figure}

\begin{figure}
\centering
\includegraphics[width=3.5in]{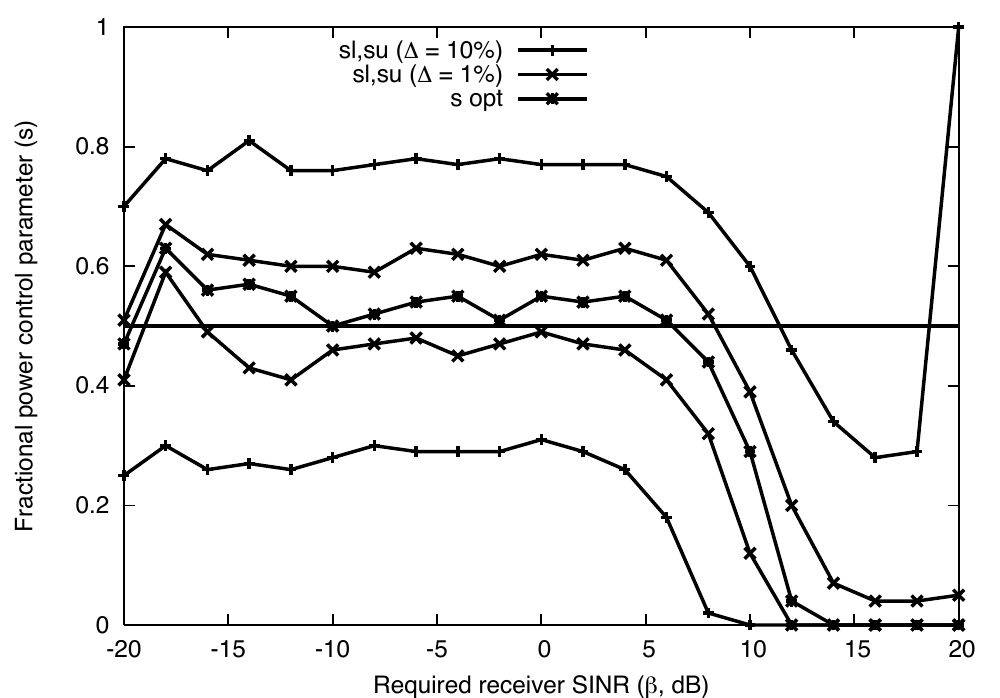}
\caption{The optimal choice of FPC exponent $s$ vs. SINR constraint
$\beta$, with $\pm 1$\% and $\pm10$\% selections for $s$.}
\label{fig:sSINR}
\end{figure}

\begin{figure}
\centering
\includegraphics[width=3.5in]{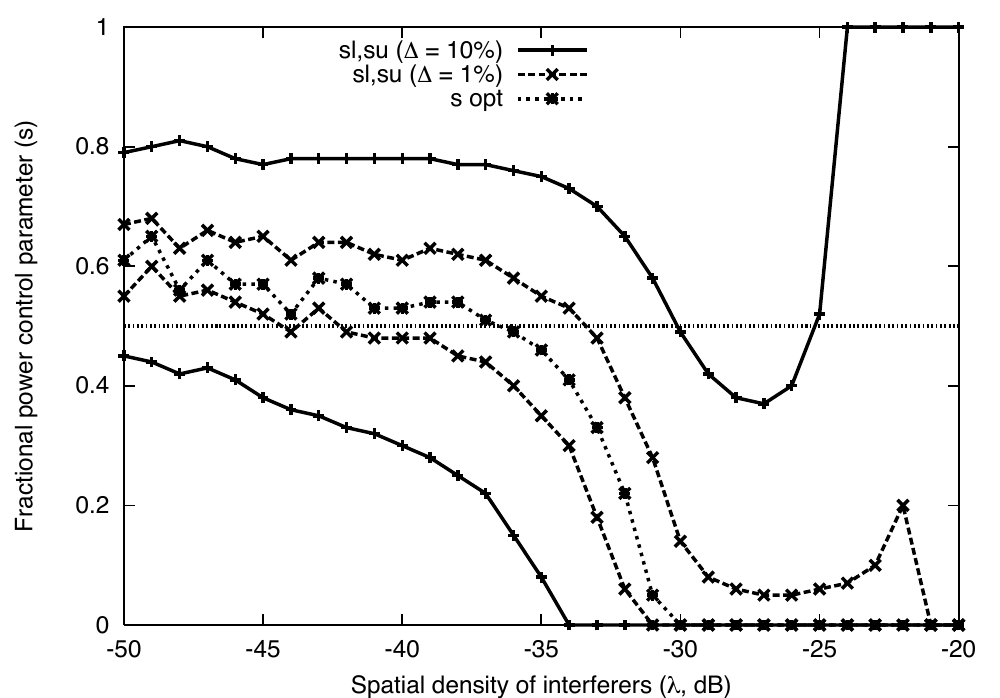}
\caption{The optimal choice of FPC exponent $s$ vs. density
$\lambda$, with $\pm 1$\% and $\pm10$\% selections for $s$.}
\label{fig:sDENSITY}
\end{figure}

\end{document}